\documentclass[11pt]{article}
\usepackage{format}
\usepackage{notation}

\title{Consistency of the Maximal Information Coefficient Estimator}

\author{%
  John Lazarsfeld \\
  Yale University \\
  john.lazarsfeld@yale.edu
  \and
  Aaron Johnson \\
  U.S. Naval Research Laboratory \\
  aaron.m.johnson@nrl.navy.mil
}

\date{July 2021}

\begin{document}
\maketitle

\begin{abstract}
  The Maximal Information Coefficient (MIC) of Reshef et al. \cite{reshef2011}
  is a statistic for measuring dependence between variable pairs in
  large datasets. In this note, we prove that MIC is a
  consistent estimator of the corresponding population statistic $\micstar$.
  This corrects an error in an argument of Reshef et al.
  \cite{reshef16-mice}, which we describe.
\end{abstract}

%
%
%
%

\section{Introduction}

The Maximal Information Coefficient (MIC) of $n$
two-dimensional data points is a statistic introduced by
Reshef et al. \cite{reshef2011} for measuring the dependence
between pairs of variables.
In later work \cite{reshef16-mice}, the authors introduced the
$\micstar$ statistic, which is defined analogously to MIC but
for jointly-distributed pairs of random variables.
Both statistics are based upon measuring the mutual information
of the discrete distributions specified by imposing finite grids
over the data (respectively, the joint distributions).

Given a dataset $D_n$ of $n$ points drawn iid from a jointly-distributed
pair of random variables $(X, Y)$, the authors of \cite{reshef16-mice}
sought to show that the $\mic$ statistic is a \textit{consistent estimator}
of $\micstar$ (i.e., that $\mic (D_n)$ converges in probability to
$\micstar(X, Y)$ as $n \to \infty$).
In this note, we identify and correct an error in an argument of
\cite{reshef16-mice} related to proving this consistency. 
Our new proof modifies the original approach of Reshef et al.,
but the result is slightly weaker than the originally-desired claim
(the set of parameters for which the consistency holds is smaller).
It is left open whether the full consistency claim of
\cite{reshef16-mice} can be recovered. 

After introducing some notation in Section~\ref{sec:prelims},
we describe the flaw in the original
argument of Reshef et al. in Section~\ref{sec:reshef-error}, and then we
provide our new proof of consistency in Section~\ref{sec:consistency}.

\section{Preliminaries}
\label{sec:prelims}

We primarily adopt the notation used by Reshef et al.~\cite{reshef16-mice},
and we summarize a few key pieces here.
For the sake of brevity, readers should refer to the original
paper of Reshef et al. for exact definitions of the $\mic$
and $\micstar$ statistics.

Let $(X, Y)$ denote a pair of jointly-distributed random variables,
and let $D_n$ be a sample of $n$ points drawn iid from $(X, Y)$.
If $G$ is a grid partition with $k \ge 2$ rows
and $\ell \ge 2$ columns, then $(X, Y)|_G$ denotes the discrete
distribution induced by $(X, Y)$ on the $k\ell$ cells of $G$.
We let $M$ denote the $\micstar$ population characteristic matrix
for $(X, Y)$, and we let $\widehat M$ denote the $\mic$ sample
characterisic matrix for $D_n$.
We say that a grid partition $\Gamma$ is an \textit{equipartition}
of $(X, Y)$ if all rows of $(X, Y)|_\Gamma$
have equal probability mass and all columns of $(X, Y)|_\Gamma$
have equal probability mass.
The total variation distance between distributions
$\Pi$ and $\Psi$ is given by
\begin{equation*}
  D_{\text{TV}}(\Pi, \Psi) = \tfrac{1}{2} \| \Pi - \Psi \|_1 .
\end{equation*}
We use $I(X, Y)$ to denote the mutual information of a
jointly-distributed pair of random variables $(X, Y)$.


%
%

\section{Flaw in Original Argument}
\label{sec:reshef-error}

Here we point out the error in the proof of Lemma 37
from \cite[Appendix A]{reshef16-mice} in the
paragraph with header ``\textit{Bounding the $\epsilon_{i, j}$}'' . 

To provide context, we consider a jointly-distributed pair of
random variables $(X, Y)$ and a dataset $D_n$ of $n$ points drawn iid
from $(X, Y)$ for some $n > 0$. 
Let $\Gamma$ be an equipartition of $(X, Y)$ with $kn^{\epsilon/4}$ rows
and $\ell n^{\epsilon/4}$ columns for some $k, \ell \ge 2$ and $\epsilon > 0$.
Let $C(n) = k\ell n^{\epsilon/2}$ denote the total number of cells in
$\Gamma$. So $\Pi = (X, Y)|_\Gamma$ and $\Psi = D_n|_\Gamma$ are discrete
distributions, and we let $\pi_{i, j}$ and $\psi_{i, j}$
denote their PMFs respectively.
Note also that because the $n$ points of $D_n$ are drawn iid from $(X, Y)$,
for any cell $(i, j)$ the quantity $n \psi_{i, j}$ is the sum
of $n$ iid Bernoulli random variables  with mean $\pi_{i,j}$,
and so $\E[n \psi_{i, j}] = n \pi_{i, j}$.

The purpose of Lemma 37 is to give a uniform bound on the
absolute difference $|I(D_n |_G) - I((X, Y)|_G)|$ for all
$k \times \ell$ grids $G$ that holds with high probability.
The strategy of the original authors is to obtain this bound
by introducing the ``common'' equipartition grid $\Gamma$.
The main consequence of the error in Lemma 37 is that
the probabilistic guarantee of the subsequent Lemma 38 does
not hold as stated, which prevents the overall argument
of consistency from going through. We start by pointing out the
error in the proof of Lemma 37 before showing an example of
its consequences on Lemma 38.

\paragraph*{Error in Lemma 37}

For every $(i, j)$, the authors define
$\epsilon_{i, j} = (\psi_{i,j} - \pi_{i, j})/\pi_{i, j}$,
and so
\begin{align}
  | \epsilon_{i, j} | = \left| \frac{\psi_{i,j} - \pi_{i, j}}{\pi_{i, j}} \right|
  \;\ge\; \delta
  \;\;\;\implies\;\;\;
  | n \psi_{i, j} - n \pi_{i, j}| \;\ge\; \delta \cdot n \pi_{i, j}.
\end{align}

Given that $n\psi_{i, j}$ is the sum of $n$ iid Bernoulli RVs,
the authors state the following multiplicative Chernoff bound:
\begin{align}
  \Pr[\; |\epsilon_{i, j}| \ge \delta \;]
  \;=\;
  \Pr[\;
  |\psi_{i,j} - \pi_{i, j} | \ge \delta \cdot \pi_{i, j}
  \;]
  &\;=\;
  \Pr[\;
  |n \psi_{i,j} - n \pi_{i, j} | \ge \delta \cdot n \pi_{i, j}
  \;] \\
  &\;\le\;
  \exp\left(
  - \Omega(n \pi_{i, j} \delta^2)
  \right).
  \label{cb-error-2}
\end{align}

For reference, we state below the standard two-sided
Chernoff bound from Corollary 4.6
in \cite{mmbook}, which says that
\begin{align}
  \Pr\left[\;
  | n \psi_{i, j}  - n \pi_{i, j}|
  \ge \delta \cdot n \pi_{i,j}
  \; \right]
  \le
  2 \exp\left(
  - \frac{n \pi_{i, j}\cdot \delta^2}{3}
  \right)
  \label{eq:chernoff}
\end{align}
for any $0 < \delta < 1$.

Now, the authors set $\delta = \sqrt{\pi_{i, j}}/ C(n)^{0.5 + \alpha}$
for some $\alpha \ge 0$.
Applying the bound from \eqref{cb-error-2} with this value of $\delta$,
the authors write
\begin{align}
  \Pr\left[\;
  |\epsilon_{i, j}|
  \;\ge\;
  \frac{\sqrt{\pi_{i, j}}}{C(n)^{0.5 + \alpha}}
  \;\right]
  \;\le\;
  \exp\left(
  - \Omega\left(
  \frac{n}{C(n)^{1 + 2\alpha}}
  \right)
  \right),
  \label{eq:eps-error-orig}
\end{align}
which incorrectly drops a dependence on $\pi_{i, j}$. 

A correct application of the true Chernoff bound from
\eqref{eq:chernoff} instead yields
\begin{align}
  \Pr\left[\;
  |\epsilon_{i, j}|
  \;\ge\;
  \frac{\sqrt{\pi_{i, j}}}{C(n)^{0.5 + \alpha}}
  \;\right]
  &=
    \Pr\left[\;
    \left| n \psi_{i,j} - n \pi_{i, j} \right|
    \;\ge\;
    \frac{\sqrt{\pi_{i, j}}}{C(n)^{0.5 + \alpha}}
    \cdot
    n \pi_{i, j}
   \;\right] \\
  &\le
  2 \exp\left(
    - \frac{n \pi_{i, j}}{3}
    \cdot
    \frac{\pi_{i, j}}{C(n)^{1 + 2\alpha}}
  \right) \\
  &\le
  2 \exp\left(
    - \Omega\left(
    \frac{n \cdot (\pi_{i, j})^2}{C(n)^{1 + 2\alpha}}
    \right)
    \right).
    \label{eq:eps-error}
\end{align}
Compared to \eqref{eq:eps-error-orig}, the
term in \eqref{eq:eps-error} resulting from the correct
application of the Chernoff bound has a dependence on
$\pi_{i, j}$.
This means that the bound on the error probability of
$|\epsilon_{i, j}| \ge \delta$ becomes non-negligible as the
value of some $\pi_{i, j}$ goes to 0.
We note that since the number of grid cells $(i, j)$ increases
with $n$ by the definition of $\Gamma$, the value of
any $\pi_{i, j}$ can decrease with $n$. 
When this occurs, we expect the bound on the probability
in \eqref{eq:eps-error} to grow undesirably large. 

\paragraph*{Ramifications on Lemma 38}

Using the corrected error probability from \eqref{eq:eps-error}
for a single $\epsilon_{i, j}$
and taking a union bound over all $(i, j)$
means that the statement of Lemma 37 now holds with
total error probability at most 
\begin{align}
  2 \sum_{(i,j)}\;
  \exp\left(- \Omega\left(
  \frac{n (\pi_{i, j})^2}{C(n)^{1 + 2\alpha}}
  \right)
  \right).
  \label{eq:37-error}
\end{align}
This updated probability seems to prevent the proof (as written) of
Lemma 38 in \cite{reshef16-mice} from working in general.

For example, consider the special case where 
$X$ and $Y$ are independent.
Then the discrete distribution $\Pi = (X, Y)|_\Gamma$ has
PMF $\pi_{i, j} = 1 / C(n)$ for all $(i, j)$ given that
$\Gamma$ is an equipartition.

The overall error term from Lemma 37 in \eqref{eq:37-error}
can be rewritten then as
\begin{align}
  2 \sum_{(i,j)}\;
  \exp\left(- \Omega\left(
  \frac{(n \cdot (1/C(n))^2}{C(n)^{1 + 2\alpha}}
  \right)
  \right)
  \;=\;
  2 \cdot C(n)
  \exp\left(- \Omega\left(
  \frac{n}{C(n)^{3 + 2\alpha}}
  \right)
  \right) . 
  \label{eq:37-error-new}
\end{align}

Now in Lemma 38, the authors consider $ k \ell \le B(n) = O(n^{1-\epsilon})$, 
and since $C(n) = k \ell n^{\epsilon/2}$, we have
\begin{align}
  C(n) \le B(n) \cdot n^{\epsilon/2} = O(n^{1 - \epsilon/2})
\end{align}
as written on page 34. 
But this means
\begin{align}
  C(n)^{3 + 2\alpha} =
  O\left(n^{(1-\epsilon/2) \cdot (3 + 2\alpha)}\right).
  \label{eq:new-38}
\end{align}
The current strategy in the proof of Lemma 38 relies on
bounding the error term in \eqref{eq:37-error-new} as
\begin{align}
  2 \cdot C(n)
  \exp\left(- \Omega\left(
  \frac{n}{C(n)^{3 + 2\alpha}}
  \right)
  \right)
  \;\le\;
  O(n)
  \exp\left(- \Omega\left(
  n^u
  \right)
  \right)
  \label{eq:38-goal}
\end{align}
for some $u > 0$. To achieve such a bound requires
$C(n)^{3+2\alpha} = o(n)$, and in turn this requires
from \eqref{eq:new-38} that
\begin{align}
  (1 - \epsilon/2)(3 + 2\alpha) =
  \left(\frac{2 - \epsilon}{2}\right)
  \left(3 + 2\alpha\right)
  \le 1.
\end{align}
Simplifying yields the constraint
\begin{align}
  2 \alpha
  \;\le\;
  \frac{2}{2- \epsilon} - 3
  \;=\;
  \frac{3\epsilon - 4}{2 - \epsilon},
\end{align}
and so to require $C(n)^{3+2\alpha} = o(n)$ means we must have
\begin{align}
  \alpha \le \frac{3\epsilon - 4}{4 - 2 \epsilon} \; .
  \label{eq:alpha-req} 
\end{align}
Now for $0 < \epsilon < 1$, which corresponds to a value of
$B(n) = O\left(n^{1 - \epsilon}\right)$ which grows with $n$,
we can verify that the right hand side of
\eqref{eq:alpha-req} is always negative.
So to obtain the desired error term in
\eqref{eq:38-goal} constrains $\alpha$ to be negative, which
contradicts the requirement of $\alpha > 0$ used in earlier
parts of the proof of Lemma 38.

So in the case where the joint distribution $(X, Y)$
is formed by two independent random variables, using the corrected
error bound from Lemma 37 in \eqref{eq:37-error-new}
renders the proof of Lemma 38 incorrect.
Thus for general joint distributions $(X, Y)$, we should
not expect the current technique in the proof of Lemma 38 to work.


%
%

\section{New Consistency Proof}
\label{sec:consistency}

We now outline an alternative approach to replace Lemmas 35-38
in \cite[Appendix A]{reshef16-mice}, which are needed to
prove the consistency of the MIC estimator in Theorem 6.

\subsection{Overview of Argument}

Our main goal is to prove a statement similar to Lemma 38 of \cite{reshef16-mice},
which probabilistically bounds the difference between corresponding
entries of $M$ and $\widehat M$:

\begin{goal}
  \label{goal:informal}
  We want to show that there exists a function
  $B(n)$ that grows with $n$ such that,
  for every joint distribution $(X, Y)$ and $n$,
  if $D_n$ is a sample of $n$ points drawn
  iid from $(X, Y)$, then 
  \begin{align*}
    | M_{k,\ell} - \widehat M_{k, \ell} | = o(1)
  \end{align*}
  holds simultaneously for all $k\ell \le B(n)$
  with probability at least $1 - o(1)$
  (where the randomness is over the sampling that
  determines $D_n$ and the asymptotics are defined
  wrt increasing $n$).
\end{goal}

If we obtain Goal~\ref{goal:informal}, then the proof
of Theorem 6 \cite[Appendix A]{reshef16-mice} (which shows
the consistency of the MIC estimator and relies
on obtaining Goal~\ref{goal:informal}) can remain unmodified.

\subsubsection{Proof Sketch of Goal~\ref{goal:informal}}

Recall that for a fixed $(k,\ell)$ pair (and assuming wlog $k \le \ell$)
we have
\begin{align}
  | M_{k, \ell} - \widehat M_{k, \ell} |
  &=
    \left|
    \; \max_{G: \; k \times \ell} \; \frac{I( (X, Y)|_G)}{\log_2 k}
    \;-\;
    \max_{G: \; k \times \ell} \; \frac{I( D_n|_G)}{\log_2 k}
    \;\right| \nonumber \\
  &=
    \frac{1}{\log_2 k} \cdot
    \left| \;
    \max_{G: \; k \times \ell} \; I( (X, Y)|_G)
    \;-\;
    \max_{G: \; k \times \ell} \; I( D_n|_G)
    \; \right| \nonumber \\
  &\le
    \max_{G: \; k \times \ell} \;
    \left| \;
    I((X, Y)|_G) \;-\; I(D_n|_G)
    \; \right| . \label{eq:goal1-max}
\end{align}

In other words, expression \eqref{eq:goal1-max} shows that
to bound the difference between $M_{k, \ell}$ and $\widehat M_{k, \ell}$,
it is sufficient to bound the maximum difference in mutual information
between the discrete distributions $(X, Y)|_G$ and $D_n|_G$ for
a grid $G$ of size at most $k \times \ell$.

So our strategy for Goal~\ref{goal:informal} is to first obtain
such a bound on expression \eqref{eq:goal1-max} for a fixed
$(k, \ell)$ that holds with probability at least $1 - p$,
where $p = o(1/ B^2(n))$.
Then by taking a union bound over all $k\ell \le B(n)$
(which is at most $B^2(n)$ pairs) 
and by choosing the function $B(n)$ appropriately, 
the statement of Goal~\ref{goal:informal} will hold with
total probability at least $1 - o(1)$. 

The purpose of the original Lemmas 35-37 in \cite{reshef16-mice}
is to bound this maximum difference in mutual information
from \eqref{eq:goal1-max} for $k \times \ell$ grids,
but here we will circumvent Lemma 37 and obtain our goal 
by adapting the original argument of Lemma 36. 
The result is a probabilistic $o(1)$ bound on
the expression \eqref{eq:goal1-max}, but we note that the
bound only holds for $k\ell \le B(n) = O(n^\alpha)$ where $0 < \alpha < 0.5$.
This is a slightly weaker guarantee compared to
the original statement of Lemma 38 and Theorem 6, which held for
$k\ell \le B(n) = O(n^\alpha)$ for $0 < \alpha < 1$.


We first state the following variant of
Lemma 36 from \cite[Appendix A]{reshef16-mice}, which follows directly
from the original proof of the lemma.

\begin{lemma}[(Variant of {\cite[Lemma 36]{reshef16-mice}})] 
  \label{lemma:36-variant}
  \hspace*{1em}
  
  \begin{itemize}
  \item
    Let $\Pi$ and $\Psi$ be random variables.
  \item
    Let $\Gamma$ be a grid with $C$ cells.
  \item
    Let $G$ be any grid with $\beta < C$ cells.
  \item
    Let $\delta$ (resp. $d$) be the total probability mass of
    $\Pi|_\Gamma$ (resp. $\Psi|_\Gamma$)
    falling in cells of $\Gamma$ that are not contained in
    individual cells of $G$.
  \item
    Let $G'$ be a sub-grid of $\Gamma$ of $\beta$ cells obtained by
    replacing every horizontal or vertical line in $G$ that is
    not in $\Gamma$ with a closest line in $\Gamma$.
  \end{itemize}
  Then
  \begin{align}
    | I(\Pi|_G) - I(\Psi|_G) |
    \;\le\;
    &\;O\left(
      \delta \log_2
      \left( \beta/\delta\right) 
      \right) \; + \label{eq:36-delta-bound} \\
    &\;O\left(
      d \log_2
      \left(\beta/d\right)
      \right) \; + \label{eq:36-d-bound} \\
    &\;| I(\Pi|_{G'}) - I(\Psi|_{G'})| . \label{eq:36-dtv-bound}
  \end{align}  
\end{lemma}

To apply this lemma, we will suppose $\Pi = (X, Y)$ and (by slight abuse of notation)
$\Psi = D_n$, and we consider any $k \times \ell$ grid $G$ where 
$k \ell \le B(n) = O(n^\alpha)$ for some $\alpha > 0$.
We will set $\Gamma$ to be an \textit{equipartition} of $\Pi$ into
$kn^{\epsilon}$ rows and $\ell n^{\epsilon}$ columns for any $\epsilon > 0$.

With these settings, we obtain a probabilistic bound on
$|I(\Pi|_G) - I(\Psi|_G)|$ that holds for every
$k \times \ell$ grid $G$ \textit{simultaneously}
(where the probability is over the randomness of the sampled points $D_n$)
by deriving probabilistic bounds on
\eqref{eq:36-delta-bound}, \eqref{eq:36-d-bound}, and \eqref{eq:36-dtv-bound}
separately and applying a union bound. 

Stated formally:

\begin{lemma}
  \label{lemma:three-bounds}
  Let $(X, Y)$ be a pair of jointly-distributed random variables
  and let $D_n$ be a dataset of $n$ points sampled iid from $(X, Y)$.
  For any $\alpha > 0$ and any $n$, consider any pair $(k, \ell)$
  where $\beta = k\ell \le B(n) = O(n^\alpha)$. 
  Let $\Gamma$ be an equipartition of $(X, Y)$ into 
  $k n^{\epsilon}$ rows, $\ell n^{\epsilon}$ columns, and
  $C(n) = k\ell n^{2\epsilon}$ total cells for any $\epsilon > 0$.
  For any grid $G$, let $G'$ be a grid of equal size as defined
  in Lemma~\ref{lemma:36-variant}. 

  Then the following probabilistic bounds hold simultaneously for
  every $k \times \ell$ grid $G$:
  \begin{enumerate}
  \item
    with probability 1:
    \begin{equation}
      \delta \le \frac{2}{n^\epsilon}
      \quad
      \implies
      \quad
      O\left(\delta \log_2 (\beta/\delta)\right)
      =
      O\left(
        \frac{\log_2n}{n^{\epsilon}}
      \right)
      \label{eq:3b-delta}
    \end{equation}
  \item
    with probability at least 
    $1 - p_d$ where $p_d := O(n^{\alpha}) \cdot e^{-\Omega(n^{1-\epsilon-\alpha})}$:
    \begin{equation}
      d \le \frac{4}{n^\epsilon}
      \quad
      \implies
      \quad
      O\left(d \log_2 (\beta/d)\right)
      =
      O\left(
        \frac{\log_2 n}{n^{\epsilon}}
      \right)
      \label{eq:3b-d}
    \end{equation}
  \item
    with probability at least
    $1 - p_{G'}$ where 
    $p_{G'} := O(n^{\alpha+2\epsilon - 3})$:
    \begin{equation}
      |I((X, Y)|_{G'}) - I(D_n|_{G'})|
      \;=\;
      O\left(
        \phi \log_2 \left(\frac{n^{\alpha}}{\phi}\right)
      \right)
      \label{eq:3b-mi}
    \end{equation}
    where
    \begin{equation}
      \phi = O\left(
        \frac{n^{\alpha + 2\epsilon} \cdot \log_2^{0.5} n}{n^{0.5}}
      \right) .
      \label{eq:3b-phi}
    \end{equation}
  \end{enumerate}
\end{lemma}

Granting Lemma~\ref{lemma:three-bounds} as true and applying
Lemma~\ref{lemma:36-variant}, we have
the following corollary that results in a bound on our original
target expression \eqref{eq:goal1-max}:

\begin{corollary}
  \label{corr:mutual-info-bound}

  Let $(X, Y)$ be a pair of jointly-distributed random variables
  and let $D_n$ be a dataset of $n$ points sampled iid from
  $(X, Y)$, and let $p_d$ and $p_G'$ be defined as in Lemma~\ref{lemma:three-bounds}.
    
  For every $0 < \alpha < 0.5$,
  there exists some $u > 0$ such that
  for all $n$ and for all $k\ell \le B(n) = O(n^\alpha)$:
  \begin{equation}
    |I((X, Y)|_G) - I(D_n|_G) |
    \;=\;
    O\left(\frac{1}{n^u}\right)
    \label{eq:corr-mi}
  \end{equation}
  holds for every $k \times \ell$ grid $G$ simultaneously
  with probability at least
  \begin{equation}
    1 - (p_d + p_{G'}) \ge 1 - O(n^{-2.5}).
    \label{eq:corr-prob}
  \end{equation}  
\end{corollary}

\begin{proof}
  As in the statement of Lemma~\ref{lemma:three-bounds},
  let $\Gamma$ be an equipartition of $(X, Y)$ into
  $kn^\epsilon$ rows and $\ell n^\epsilon$ columns.
  When $0 < \alpha < 0.5$, any choice of $0 < \epsilon < 1/4 - \alpha/2$
  ensures $\alpha + 2\epsilon < 0.5$, which means that
  expressions \eqref{eq:3b-delta}, \eqref{eq:3b-d}, and
  \eqref{eq:3b-mi} from Lemma~\ref{lemma:three-bounds} are
  all $O(n^{-u})$ for some positive constant $u < 0.5 - (\alpha + 2 \epsilon)$.
  So for every $k\ell \le B(n) = O(n^{\alpha})$,
  expression~\eqref{eq:corr-mi} of the corollary follows from applying
  Lemma~\ref{lemma:36-variant} with these three bounds.
  The same setting of $\alpha$ and $\epsilon$ also means that
  $p_{G'} = O(n^{\alpha + 2\epsilon - 3}) = O(n^{-2.5})$ and
  $p_{d} = O(n^{\alpha}) \cdot e^{-\Omega(n^{1-\epsilon-\alpha})} = O(n^{-2.5})$
  (since $p_{d}$ is $o(n^{-c})$ for any $c > 0$), from which
  expression~\eqref{eq:corr-prob} of the corollary follows. 
\end{proof}

We can now use Corollary~\ref{corr:mutual-info-bound} to formally
state the following theorem which achieves our original
Goal~\ref{goal:informal}.

\begin{theorem}
  \label{thm:final-goal1}
  Let $(X, Y)$ be a pair of jointly-distributed random variables
  and let $D_n$ be a dataset of $n$ points sampled iid from
  $(X, Y)$.
  For every $0 < \alpha < 0.5$,
  there exists a constant $u > 0$ such that
  for all $n$:
  \begin{equation}
    | M_{k,\ell} - \widehat M_{k, \ell} | = O\left(\frac{1}{n^u}\right)
    \label{eq:final-kl-bound}
  \end{equation}
  holds for every $(k, \ell)$ pair where
  $k\ell \le B(n) = O(n^\alpha)$ simultaneously
  with probability at least $1 - O(n^{-1.5})$
  (where the randomness is over the sampling that
  determines $D_n$). 
\end{theorem}

\begin{proof}
  Recall expression~\eqref{eq:goal1-max}, which says
  that
  \begin{align*}
    | M_{k, \ell} - \widehat M_{k, \ell} |
     \;\le\;
      \max_{G: \; k \times \ell} \;
      \left| \;
      I((X, Y)|_G) \;-\; I(D_n|_G)
      \; \right|
  \end{align*}
  for a fixed pair $(k, \ell)$.
  Then by Corollary~\ref{corr:mutual-info-bound}, for every $0 < \alpha < 0.5$
  and every $n$, 
  there exists some $u > 0$ such that the right hand side of this expression
  is $O(n^{-u})$ with probability at least $1 - O(n^{-2.5})$
  for every $(k, \ell)$ where $k\ell \le B(n) = O(n^\alpha)$.
  Given that there are at most $O(n^{2\alpha})$ pairs satisfying
  $k\ell \le O(n^\alpha)$, it follows from a union bound that
  $|M_{k, \ell} - \widehat M_{k, \ell}| = O(n^{-u})$ for
  \textit{all} such $(k, \ell)$ pairs simultaneously with probability
  $1 - O(n^{2\alpha - 2.5}) \ge 1 - O(n^{-1.5})$. \qedhere
\end{proof}

This gives us the desired result of Goal~\ref{goal:informal},
and it now remains to prove the three parts of
Lemma~\ref{lemma:three-bounds}.


%
%

\subsection{Lemma~\ref{lemma:three-bounds} Proof: Parts 1 and 2}

Recall that $\Pi = (X,Y)$, $\Psi = D_n$, $\Gamma$ is an equiparition
of $\Pi$ with $kn^{\epsilon}$ rows and $\ell n^{\epsilon}$ columns,
and $G$ is a $k \times \ell$ grid where $k\ell \le B(n) = O(n^\alpha)$
for some $\alpha > 0$.
We define $\delta$ (resp. $d$) to be the total mass of
$\Pi|_\Gamma$ (resp. $\Psi|_\Gamma$)
falling in cells of $\Gamma$ that are not contained in
individual cells of $G$.

We will prove parts 1 and 2 of Lemma~\ref{lemma:three-bounds}
together, which say that:
\begin{enumerate}
\item
  $\delta \le (2 / n^{\epsilon})$ with probability 1.
\item
  $d \le (4/ n^{\epsilon})$ with probability at least
  $1 -  p_d$, where $p_d :=  O(n^\alpha) \cdot e^{- \Omega(n^{1-\epsilon-\alpha})}$.
\end{enumerate}
Our strategy will be to bound $\delta$ by $\delta '$, and then
to show $d \le 2\delta'$ with probability all but $p_d$. 

\subsubsection*{Chernoff Bounds \cite[Chapter 4]{mmbook}}

First, we (re)state two standard Chernoff bounds that will be used
in this section and the next:

Let $X = X_1 + \dots X_n$, where each $X_i$ is an iid Bernoulli RV
with $\E[X_i] = \mu$.
\begin{enumerate}[i]
\item
  \textit{Two-sided tail bound}: for any $0 < t < 1$:
  \begin{align}
    \Pr\left[ \;
    | X - n\mu | \ge t \cdot n\mu
    \; \right]
    \;\le\;
    2 \cdot \exp\left(
    - \frac{n\mu t^2}{3}
    \right) \label{eq:chernoff-double}
  \end{align}
\item
  \textit{Upper tail bound}: for any $0 < t \le 1$ and
  $\hat \mu \ge \mu = \E[X_i]$: 
  \begin{align}
    \Pr\left[ \;
     X \ge (1 + t) \cdot n \hat \mu
    \; \right]
    \;\le\;
    \exp\left(
    - \frac{n\hat\mu t^2}{3}
    \right) \label{eq:upper}
  \end{align}
\end{enumerate}

\subsubsection*{Bound on $\delta$}

By definition, $\delta$ is the sum of mass in
a subset of columns and rows of $\Pi|_\Gamma$.
Let $\pi_{i, j}$ denote the pmf at cell $(i, j)$ of
$\Pi|_\Gamma$, let $\pi_{i, *}$ denote the
total mass of $\Pi|_\Gamma$ in row $i$, and let
$\pi_{*, j}$ denote the total mass of $\Pi|_\Gamma$ in column $j$.
So
\begin{align*}
  \pi_{*, j} &= \sum_{i} \pi_{i, j} = \frac{1}{\ell n^{\epsilon}} \\
  \pi_{i, *} &= \sum_{j} \pi_{i, j} = \frac{1}{kn^{\epsilon}}
\end{align*}
by the definition of $\Gamma$ as an equipartition of $\Pi$.
Now let $K$ be the column indices of $\Gamma$ containing
a column separator of $G$, and let $R$ be the row indices
of $\Gamma$ containing a row separator of $G$.
Since $G$ is a $k \times \ell$ grid, we must have
$|K| \le \ell$ and $|R| \le k$. Then (with probability 1):
\begin{align*}
  \delta
  &\le \sum_{j \in K} \pi_{*, j} + \sum_{i \in R} \pi_{i, *} \\
  &\le \ell \cdot \pi_{*, j} + k \cdot \pi_{i, *} 
    \;=\; \frac{\ell}{\ell n^{\epsilon}} + \frac{k}{kn^{\epsilon}}
    \;=\; \frac{2}{n^\epsilon}. 
\end{align*}

\subsubsection*{Bound on $d$ with probability $1 - p_d$}

Again by definition, $d$ is the sum of mass in a subset
of columns and rows of $\Psi|_\Gamma = D_n|_\Gamma$.
We let $\psi_{i, j}$ denote the pmf at cell $(i, j)$
of $\Psi|_\Gamma$, and we define $\psi_{*, j}$ and
$\psi_{i, *}$ analogously to $\pi_{*, j}$ and $\pi_{i, *}$. 
We will show that each $\psi_{*, j} \le 2 \pi_{*, j}$
(respectively $\psi_{i, *} \le 2 \pi_{i, *}$)
probabilistically.

Observe that $n \cdot \psi_{*, j}$ is a sum of $n$ iid Bernoullis,
each with mean $\pi_{*, j}$. So
\begin{equation}
  \E[n \cdot \psi_{*,j}]
  \;=\; n \cdot \pi_{*, j}
  \;=\; \frac{n}{\ell n^{\epsilon}}
  \;=\; \frac{n^{1-\epsilon}}{\ell}.
\end{equation}
Then by the Chernoff bound \eqref{eq:upper}:
\begin{align*}
  \Pr\left[
  n \cdot \psi_{*, j}
  \ge
  2 \cdot (n^{1-\epsilon}/\ell)
  \right]
  &\le
    \exp\left(
    - \frac{n^{1-\epsilon}}{3\ell}
    \right) \\
  &\le
    \exp\left(
    - \Omega\left(n^{1-\epsilon-\alpha}\right)
    \right), 
\end{align*}
where the final inequality is due to  $\ell = O(n^\alpha)$,
which follows from the assumption that $k \ell = O(n^\alpha)$. 
So for each $j \in K$ we have $\psi_{*, j} \le 2 \pi_{*, j}$
with probability all but $e^{- \Omega(n^{1-\epsilon - \alpha})}$.
A similar calculation shows that
$\psi_{i, *} \le 2 \pi_{i, *}$ for each $i \in R$
with probability all but $e^{- \Omega(n^{1-\epsilon - \alpha})}$.

Combining the two inequalities and taking a union bound shows
\begin{align*}
  d
  &\le \sum_{j \in K} \psi_{*, j} + \sum_{i \in R} \psi_{i, *} \\
  &\le \sum_{j \in K} 2 \pi_{*, j} + \sum_{i \in R} 2 \pi_{i, *}
  \;\le\; 2 \left(\frac{2}{n^{\epsilon}}\right)
    \;=\; \frac{4}{n^\epsilon}
\end{align*}
with probability all but 
$p_d := O(n^\alpha)\cdot e^{-\Omega(n^{1-\epsilon- \alpha})}$,
since $|K| + |R| \le k + \ell \le k\ell \le O(n^\alpha)$
and by using the bound on $\delta$ previously established.


%
%

\subsection{Lemma~\ref{lemma:three-bounds} Proof: Part 3}

Recall that given the grid $\Gamma$ (which is an equipartition
of $(X, Y)$ into $kn^\epsilon$ rows and $\ell n^\epsilon$ columns)
and the $k \times \ell$ grid $G$, the grid $G'$ is a $k\times\ell$
sub-grid of $\Gamma$ obtained by replacing every horizontal or
vertical line in $G$ that is not in $\Gamma$ with a closet line in $\Gamma$.

To prove an upper bound on the quantity
$| I((X, Y)|_{G'}) - I(D_n|_{G'}) |$, we will use
Proposition 40 from Appendix B of 
\cite{reshef16-mice}, which relates the 
statistical distance between two discrete distributions
to their change in mutual information:

\begin{proposition}[({\cite[Proposition 40, Appendix B]{reshef16-mice}})]
  \label{prop:40}
  Let $\Pi$ and $\Psi$ be discrete distributions over $k \times \ell$ grids. 
  If $D_{TV}(\Pi, \Psi) \le \delta$ for any $0 < \delta \le 1/4$, then
  \begin{align*}
  |I(\Pi) - I(\Psi)|
  \le
    O\left(
    \delta \log_2 \left(\frac{\min\{k, l\}}{\delta}\right)
    \right) . 
  \end{align*}
\end{proposition}

Because $D_{TV}(\Pi, \Psi) = \tfrac{1}{2} \| \Pi - \Psi \|_1$,
and since $G'$ is a subgrid of $\Gamma$, it follows by the triangle
inequality that
\begin{equation*}
  D_{TV}((X, Y)|_{G'}, D_n|_{G'} )
  \;\le\;
  D_{TV}((X, Y)|_{\Gamma}, D_n|_{\Gamma}).
\end{equation*}

Thus if we obtain a bound $D_{TV}((X, Y)|_{\Gamma}, D_n|_{\Gamma}) \le \phi$,
then applying Proposition~\ref{prop:40} yields
\begin{equation*}
  | I((X, Y)|_{G'}) - I(D_n|_{G'}) |
  = 
  O\left(
    \phi \log_2 \left(\frac{\min\{k, \ell\}}{\phi}\right)
  \right)
  =
  O\left(
    \phi \log_2 \left(\frac{n^{\alpha}}{\phi}\right)
  \right)
\end{equation*}
by the assumption that $k\ell \le B(n) = O(n^{\alpha})$.

So given $(X, Y)$, the dataset $D_n$, and the equipartition $\Gamma$
of $C(n) = k\ell n^{2\epsilon}$ total cells,
we will prove the following bound on
$D_{TV}((X, Y)|_{\Gamma}, D_n|_{\Gamma})$,
which implies Part (3) of Lemma~\ref{lemma:three-bounds}. 

\begin{lemma}
  \label{lemma:dtv-bound}
  Let $(X, Y)$ be a pair of jointly-distributed random variables
  and let $D_n$ be a dataset of $n$ points sampled iid from $(X, Y)$.
  For any $\alpha > 0$ and any $n$, consider any pair $(k, \ell)$
  where $k\ell \le B(n) = O(n^\alpha)$. 
  Let $\Gamma$ be an equipartition of $(X, Y)$ into 
  $k n^{\epsilon}$ rows, $\ell n^{\epsilon}$ columns, and
  $C(n) = k\ell n^{2\epsilon}$ total cells for any $\epsilon > 0$.
  Then
  \begin{equation*}
    D_{TV}((X, Y)|_{\Gamma}, D_n|_{\Gamma})
    =
    O\left(
      \frac{n^{\alpha + 2\epsilon}\cdot \log_2^{0.5}n}{n^{0.5}}
    \right)
  \end{equation*}
  with probability at least $1 - O(n^{\alpha + 2\epsilon - 3})$. 
\end{lemma}

\begin{proof}
  Given $(X, Y)$, a sample $D_n$ of $n$ points drawn iid from $(X, Y)$,
  and the grid $\Gamma$, define the discrete distributions
  \begin{align*}
    \Pi &= (X, Y)|_\Gamma \;\; \text{with pmf} \;\; \pi_{i, j} \\
    \Psi &= D_n|_\Gamma \;\; \text{with pmf} \;\; \psi_{i, j}
  \end{align*}
  where $i \in [kn^{\epsilon}]$ and $j \in [\ell n^{\epsilon}]$
  (note that the use of $\Pi$ and $\Psi$ here differs slightly from
  the previous subsection).
  Also, for every cell $(i, j)$ of $\Gamma$, we say that
  \begin{align*}
    &(i, j) \;\; \text{is \textit{large} if} \;\; \pi_{ij} > \frac{9 \log_2 n}{n}\\
    \text{and}\;\;
    &(i, j) \;\; \text{is \textit{small} if} \;\; \pi_{ij} \le \frac{9 \log_2 n}{n},
  \end{align*}
  and let $L$ and $S$ denote the sets of \textit{large} and \textit{small}
  $(i, j)$ cells, respectively.

  Now recall that
  \begin{align}
    D_{TV}(\Pi, \Psi)
    \;\le\;
    \| \; \Pi - \Psi \; \|_1
    &\;=\;
      \sum_{(i, j)} | \pi_{i,j} - \psi_{i, j} | \nonumber \\
    &\;=\;
      \sum_{(i, j) \in L} | \pi_{i, j} - \psi_{i, j} |
      \;+\;
      \sum_{(i, j) \in S} | \pi_{i, j} - \psi_{i, j} | . \label{eq:dtv-1}
  \end{align}
  By the triangle inequality, we have that
  \begin{align*}
    \sum_{(i, j) \in S} | \pi_{i, j} - \psi_{i, j} |
    &\;\le\;
    \sum_{(i, j) \in S} | \pi_{i, j}| +  | \psi_{i, j} | \\
    &\;=\;
      \sum_{(i, j) \in S} | \pi_{i, j}| +
      \sum_{(i, j) \in S} | \psi_{i, j} |,
  \end{align*}
  and substituting back into \eqref{eq:dtv-1} gives
  \begin{align}
    D_{TV}(\Pi, \Psi)
    \;\le\;
    \sum_{(i, j) \in L} | \pi_{i, j} - \psi_{i, j} |
    \;+\;
    \sum_{(i, j) \in S} | \pi_{i, j}|
    \;+\;
    \sum_{(i, j) \in S} | \psi_{i, j} | . \label{eq:dtv-2}
  \end{align}
  So to bound $D_{TV}(\Pi, \Psi)$, we will bound
  each term of \eqref{eq:dtv-2} separately. \\

  \textbf{Bound on $\sum | \pi - \psi|$ for large $(i, j)$}: 
  
  Observe that $\psi_{i, j}$ is the fraction of
  points of $D_n$ contained in cell $(i, j)$ of $\Gamma$.
  Each point has probability $\pi_{i, j}$ of falling in
  cell $(i, j)$,
  so $n \cdot \psi_{i, j}$ is the sum of $n$ iid Bernoullis, 
  each with mean $\pi_{i, j}$.

  Using the two-sided Chernoff bound from \eqref{eq:chernoff-double},
  we then have that for any $(i, j)$
  \begin{align}
    \Pr\left[
    | \; n\psi_{i,j} - n\pi_{i, j} \; | \ge t n \pi_{i, j}
    \right]
    \;\le\;
    2 \cdot \exp\left(
    - \frac{n\pi_{i, j} \cdot t^2}{3}
    \right) \label{eq:cb-large-1} 
  \end{align}
  for any $0 < t < 1$.

  Since $\pi_{i, j} > (9 \log_2 n)/n$ for each \textit{large} $(i, j)$,
  observe that setting $t = 3\sqrt{\frac{\log_2 n}{n \cdot \pi_{i, j}}}$
  means that
  \begin{align*}
    t \;<\;
    \frac{3\sqrt{\log_2  n}}{\sqrt{n}} \cdot \frac{\sqrt{n}}{\sqrt{9 \log_2 n}} = 1,
  \end{align*}
  and thus the bound in \eqref{eq:cb-large-1} can be applied\footnote{%
    The only lower bound constraint on $\pi$ for large $(i,j)$
    comes from ensuring $t < 1$
    so that our Chernoff bound variant can be applied. 
    This constraint also determines the $\sqrt{n}$ denominator in
    \eqref{eq:large-bound}.}.

  Then for each \textit{large} $(i, j)$, using this setting of
  $t = 3\sqrt{\frac{\log_2 n}{n \cdot \pi_{i, j}}}$ gives
  \begin{align*}
    \Pr\left[
    | \; n\psi_{i,j} - n\pi_{i, j} \; | \ge t n \pi_{i, j}
    \right]
    &\;\le\;
    2 \cdot \exp\left(
    - \frac{n\pi_{i, j} \cdot t^2}{3}
    \right) \\
    &\;\le\;
      2 \cdot \exp\left(
      - 3 \log_2 n
      \right) 
    \;\le\;
      \frac{2}{n^3}.
  \end{align*}
  So for each \textit{large} $(i, j)$, with probability
  at least $1 - 2n^{-3}$ we have 
  \begin{align*}
    |\; n\pi_{i, j} - n\psi_{i, j} \;|
    \;\le\;
    t n \pi_{i, j}
    \;\;\Longleftrightarrow\;\;
    |\; \pi_{i, j} - \psi_{i, j} \;|
    \;\le\;
    t \pi_{i, j}
  \end{align*}
  which means by our setting of $t$ that
  \begin{align*}
    |\; \pi_{i, j} - \psi_{i, j} \;|
    \;\le\;
    \frac{3 \log_2^{0.5} n}{\sqrt{\pi_{i, j} n}} \cdot \pi_{i, j} 
    \;=\;
       \frac{3\log_2^{0.5} n}{\sqrt{n}} \cdot \sqrt{\pi_{i, j}}
    \;\le\;
      \frac{3\log_2^{0.5} n}{\sqrt{n}}.
  \end{align*}
  Here, the last inequality holds given that $\pi_{i, j} \le 1$
  for any \textit{large} $(i, j)$.

  Now summing over all \textit{large} $(i, j)$ and taking
  a union bound, we have
  \begin{align}
    \sum_{(i, j) \in L} | \; \pi_{i, j} - \psi_{i, j} \; |
    \;\le\;
    k\ell n^{2\epsilon} \cdot \frac{3\log_2^{0.5} n}{n^{0.5}}
    \;\le\;
    O\left(
    \frac{n^{\alpha + 2\epsilon} \cdot \log_2^{0.5}n}{n^{0.5}}    
    \right)
    \label{eq:large-bound}
  \end{align}
  with probability at least $1 - O(n^{\alpha + 2\epsilon - 3})$,
  since $|L| \le C(n) = k\ell n^{2\epsilon}$
  and by the assumption that $k\ell \le B(n) = O(n^\alpha)$. \\

  \textbf{Bound on $\sum | \pi | $ for small $(i, j)$}:

  Recall that cell $(i, j)$ is \textit{small} if
  $\pi_{i, j} \le \frac{9 \log_2 n}{n}$, and so
  with probability 1:
  \begin{align*}
    \sum_{(i, j) \in S} \pi_{i, j}
    \;\le\;
    \frac{9 \cdot C(n) \log_2 n}{n} . 
  \end{align*}

  \textbf{Bound on $\sum | \psi | $ for small $(i, j)$}:
    
  Observe that
  $n \cdot \left(\sum_{(i, j) \in S} \psi_{i, j}\right)$
  is the total number of points of $D_n$ contained in
  \textit{small} $(i, j)$ cells of $\Gamma$ and is thus
  the sum of $n$ iid Bernoullis, each with mean
  $\sum_{(i, j) \in S} \pi_{i, j}$.

  So in expectation we have
  \begin{align*}
    \E\left[\;
    n \cdot \left(\sum_{(i, j) \in S} \psi_{i, j}\right)
    \;\right]
    &\;=\;
    n \cdot \left(\sum_{(i, j) \in S} \pi_{i, j}\right) \\
    &\;\le\;
      n \cdot
      \left( \frac{9 \cdot C(n) \log_2 n}{n} \right)
      \;=\;
      9 \cdot C(n) \log_2 n,
  \end{align*}
  where the inequality is due to the bound on $\sum | \pi|$
  for \textit{small} $(i, j)$ from the previous step.
  
  Now using the upper Chernoff bound from \eqref{eq:upper} and
  setting $t = 1$ gives
  \begin{align*}
    \Pr\left[\;
    n \cdot \left(\sum_{(i, j) \in S} \psi_{i, j}\right)
    \;\ge\;
    18 \cdot C(n) \log_2 n
    \;\right]
    &\;\le\;
      \exp\left(
      - 3 k \ell n^{2\epsilon} \cdot \log_2 n
      \right) \\
    &\;\le\;
      \exp\left(
      - \Omega(n^{2\epsilon})
      \right)
  \end{align*}
  since $C(n) = k\ell n^{2\epsilon}$ and $k, \ell \ge 2$.

  This means that with probability at least
  $1 - e^{-\Omega(n^{2\epsilon})}$ we have
  \begin{align*}
    n \cdot \left(\sum_{(i, j) \in S} \psi_{i, j}\right)
    \;<\; 
    18 \cdot C(n) \log_2 n
  \end{align*}
  and so
  \begin{align*}
    \sum_{(i, j) \in S} \psi_{i, j}
    \;<\;
    18 \cdot \frac{C(n) \log_2 n}{n}
    \;=\;
    O\left(\frac{n^{\alpha + 2\epsilon} \cdot \log_2 n}{n}\right) .
  \end{align*} 

  \textbf{Final bound on $D_{TV}(\Pi, \Psi)$}

  To conclude the proof, using the preceding three individual bounds
  on the terms from \eqref{eq:dtv-2} we have that  
  \begin{align*}
    D_{TV}(\Pi, \Psi)
    &\;\le\;
      \sum_{(i, j) \in L} | \pi_{i, j} - \psi_{i, j} |
      \;+\;
      \sum_{(i, j) \in S} | \pi_{i, j}|
      \;+\;
      \sum_{(i, j) \in S} | \psi_{i, j} | \\
    &\;\le\;
      O\left(
      \frac{n^{\alpha + 2\epsilon} \cdot \log_2^{0.5} n}{n^{0.5}}
      \right)
      \;+\;
      O\left(
      \frac{n^{\alpha + 2\epsilon} \cdot \log_2 n}{n}
      \right)
      \;+\;
      O\left(
      \frac{n^{\alpha + 2\epsilon} \cdot \log_2 n}{n}
      \right) \\
    & \;=\;
      O\left(
      \frac{n^{\alpha + 2\epsilon} \cdot \log_2^{0.5} n}{n^{0.5}}
      \right)
  \end{align*}
  with probability at least
  $1 - O(n^{\alpha + 2\epsilon - 3}) - e^{- \Omega(n^{2\epsilon})}
  \ge 1 - O(n^{\alpha + 2\epsilon - 3})$.
\end{proof}


%
%

\section{Conclusion}
\label{sec:conclusion}

We emphasize that Theorem~\ref{thm:final-goal1}, which
is the core new result needed to prove the
consistency of the $\mic$ estimator, yields a slightly
weaker result than in the original claim of \cite{reshef16-mice}.
In our theorem, we show that the difference between
the corresponding $(k, \ell)$ entries of $\widehat M$ and $M$
is small when $k\ell \le B(n) = O(n^\alpha)$
for $0 < \alpha < 0.5$ with high probability. 
The original claim of Reshef et al. sought to prove
the same result for $k \ell \le B(n) = O(n^\alpha)$ for
$0 < \alpha < 1$.
We suspect that the statement of Theorem~\ref{thm:final-goal1}
\textit{does} hold for $0 < \alpha < 1$, but that the techniques
used here are insufficient to prove this stronger claim.


\paragraph{Acknowledgements}
We would like to thank Yakir Reshef for his
time and help in discussing his original work.
We would also like to thank Joan Feigenbaum for her helpful guidance.
This work has been supported by the Office of Naval Research.

\bibliographystyle{alpha}
\bibliography{references.bib}

\newcommand{\etalchar}[1]{$^{#1}$}
\begin{thebibliography}{RRF{\etalchar{+}}16}

\bibitem[MU05]{mmbook}
Michael Mitzenmacher and Eli Upfal.
\newblock {\em Probability and Computing: Randomized Algorithms and
  Probabilistic Analysis}.
\newblock Cambridge University Press, 2005.

\bibitem[RRF{\etalchar{+}}11]{reshef2011}
David~N Reshef, Yakir~A Reshef, Hilary~K Finucane, Sharon~R Grossman, Gilean
  McVean, Peter~J Turnbaugh, Eric~S Lander, Michael Mitzenmacher, and Pardis~C
  Sabeti.
\newblock Detecting novel associations in large data sets.
\newblock {\em Science}, 334(6062):1518--1524, 2011.

\bibitem[RRF{\etalchar{+}}16]{reshef16-mice}
Yakir~A Reshef, David~N Reshef, Hilary~K Finucane, Pardis~C Sabeti, and Michael
  Mitzenmacher.
\newblock Measuring dependence powerfully and equitably.
\newblock {\em The Journal of Machine Learning Research}, 17(1):7406--7468,
  2016.

\end{thebibliography}

\end{document}